\documentclass[aap]{imsart}

\startlocaldefs
\endlocaldefs
\usepackage{fancyhdr}
\usepackage{amsfonts}
\usepackage{times}
\usepackage[pdftex]{graphicx}
\usepackage{mathtools}

\usepackage{longtable}
\usepackage{lipsum}

\usepackage{amssymb}
\usepackage{amsmath}
\usepackage{cite}
\usepackage{verbatim}
\usepackage{braket}
\usepackage{subfigure}
\newtheorem{thm}{Theorem}[section]
\usepackage{epsfig}
 \usepackage{epstopdf}
 \usepackage{psfrag}
 \usepackage{color}
\graphicspath{{../}}
 
 \usepackage{subfigure}
\usepackage{amsmath,amssymb}
\usepackage{color}
\usepackage{cite}
\usepackage{multirow,tabularx}
\usepackage{ifthen}
\usepackage{times}
\usepackage{amsbsy} 
\usepackage{epsfig}
\usepackage{color}

\usepackage{stackrel}

\usepackage{amsfonts}
\usepackage{stfloats}

\usepackage[all]{xy}

\usepackage{makeidx} 
\makeindex 

\newtheorem{defn}[thm]{Definition}

\newtheorem{remark}[thm]{Remark}

\newenvironment{proof}{ \textbf{Proof:} }{ \hfill $\Box$}

\def\bb0{{\mathbb{0}}}
\def\b1{{\mathbf{1}}}

\def\bb{{\mathbf{b}}}

\def\b0{{\mathbf{0}}}

\def\bbE{{\mathbb{E}}}

\def\bbP{{\mathbb{P}}}

\def\bbR{{\mathbb{R}}}

\def\bbZ{{\mathbb{Z}}}

\def\cG{\mathcal{G}}

\def\sfN{\mathsf{N}}

\def\sf0{{\mathsf{0}}}

\def\1{{\bf 1}}

\def\SINR{{\mathsf{SINR}}}

\def\nn{\nonumber}

\newcommand{\helv}{%
    \fontfamily{phv}\fontseries{b}\fontsize{9}{11}\selectfont}
{
   \end{minipage}
   \vspace*{\stretch{3}}
   \clearpage
}

\usepackage{subfigure}
\usepackage{amsmath,amssymb}
\usepackage{graphicx}
\usepackage{color}
\usepackage{cite}
\usepackage{multirow,tabularx}
\usepackage{ifthen}
\usepackage{times}
\usepackage{graphicx}
\usepackage{amsbsy}
\usepackage{epsfig}
\usepackage{color}

\usepackage{stackrel}

\usepackage{amsfonts}
\usepackage{stfloats}

\usepackage[all]{xy}
\newcommand{\be}{\begin{equation}}
\newcommand{\ee}{\end{equation}}
\newcommand{\bea}{\begin{eqnarray}}
\newcommand{\beas}{\begin{eqnarray*}}

\newcommand{\eea}{\end{eqnarray}}
\newcommand{\eeas}{\end{eqnarray*}}

\newcommand{\remove}[1]{}

\def\om{\omega}

\newcommand{\ep}{\epsilon}
\newcommand{\al}{\alpha}

\newcommand{\del}{\delta}

\def\ga{\gamma}

\def\th{\theta}

\def\1{\mathbf{1}}

\def\f{\frac}

\begin{document}

\begin{frontmatter}

\title{Achieving Non-Zero Information Velocity in Wireless Networks\protect}
%
\runtitle{Non-Zero Information Velocity in PPP}
%
\begin{aug}
  \author{\fnms{Srikanth}  \snm{Iyer}\corref{}\thanksref{t2}\ead[label=e1]{skiyer@math.iisc.ernet.in}},
  \author{\fnms{Rahul} \snm{Vaze}\thanksref{t3}\ead[label=e2]{vaze@tcs.tifr.res.in}}
  \thankstext{t2}{Research Supported in part by UGC center for advanced studies.}
  \thankstext{t3}{Research Supported in part by INSA young scientist award grant.}
  \runauthor{Iyer and Vaze}
  \affiliation{Indian Institute of Science, Bangalore and Tata Institute of Fundamental Research, Mumbai}
%
%
\end{aug}
\begin{abstract} In wireless networks, where each node transmits independently of other nodes in the network (the ALOHA protocol), the expected delay experienced
by a packet until it is successfully received at any other node is known
to be infinite for signal-to-interference-plus-noise-ratio (SINR) model with node locations
distributed according to a Poisson point process. Consequently,
the information velocity, defined as the limit of the ratio of the
distance to the destination and the time taken for a packet to
successfully reach the destination over multiple hops, is zero, as
the distance tends to infinity. A nearest neighbor distance based
power control policy is proposed to show that the expected delay
required for a packet to be successfully received at the nearest
neighbor can be made finite. Moreover, the information velocity is
also shown to be non-zero with the proposed power control policy. The condition under which these results hold does not depend on the intensity of the underlying Poisson point process.
\end{abstract}


\end{frontmatter}

\section{Introduction}
Typically, nodes in a wireless network are separated by large distances and packets are routed from source to their destination via many other nodes or over multiple hops. Therefore, to understand the connectivity or  information flow in a wireless network, a space-time SINR graph is considered. Such a graph models the evolution of the spatial as well as the temporal connections in the network. The space-time SINR graph is a directed and weighted multigraph that represents the most complete random graph model for wireless networks \cite{haenggibook}. The SINR (signal-to-interference-plus-noise-ratio)  is a ratio of the relative strength of the intended signal and the undesirable interference from simultaneously active unintended nodes of the wireless network. The SINR between any two nodes evolves with time and depends not only on the distance between the two nodes but also on the location of the other nodes in the network. At any time, a directional connection is established from a node at $x$ to another node at $y$ if the SINR from $x$ to $y$ is larger than a threshold. Such a connection represents the ability of node $x$ to deliver meaningful information to $y$.

Let $\Phi \subset \bbR^2$ be a point process that specifies the location of the nodes of the network. For any $t  \in \bbZ_+$, let $\Phi_T(t) \subset \Phi$ be the set of nodes that are transmitting at time $t$ and $\Phi_R(t) = \Phi \setminus \Phi_T(t)$ be the set of nodes in receiving mode at time $t$. Formally, the SINR from a node $x \in \Phi_T(t)$ to a node $y \in \Phi_R(t)$ is given by
\begin{equation}\label{eq:SINR1}
\SINR_{xy}(t) := \frac{P_x(t) h_t(x,y) \ell(x,y) }{ \ga \sum_{z \in \Phi_T(t) \backslash \{x,y\}} P_z(t)
h_t(z,y) \ell(z,y) + \sfN},
\end{equation}
where $\ell (.,.)$ is the distance based signal attenuation or {\it path-loss} function, $P_x(t)$ is the transmitted power from $x$ at time $t$, $\gamma$ is the interference suppression constant, $h_t(u,v),$ $u,v \in \Phi$, are the space-time fading coefficients that model the loss (or gain) from node $u$ to $v$ at time $t$ due to signal propagation via a wireless medium, and $\sfN$ is the variance of the so-called additive white Gaussian noise. By an abuse of notation, we will often use $\ell(|x-y|)$ for $\ell(x,y)$, since the path loss is a function of the distance between $x$ and $y$. The term $\sum_{z \in \Phi_T(t) \backslash \{x,y\}} P_t(z) h_t(z,y) \ell(z,y) $ in the denominator of \eqref{eq:SINR1} is referred to as the interference. Note that we do not include the nodes at $x,y$ in the interference term since transmission from node $x$ is the signal of interest and node $y$ is in receiving mode. $\SINR_{xy}(t)$ is set to be zero at time $t$ if either node $x \in \Phi_R(t)$ or if node $y \in \Phi_T(t)$. Define the indicator random variables
\begin{equation}
e_{xy}(t) := \left\{\begin{array}{ll} 1 & \mbox{if } \SINR_{xy}(t) > \beta, \\ 0 & \mbox{otherwise,}
\end{array}\right.,
\end{equation}
where $\beta > 0$ is arbitrary. The space-time SINR graph is defined to be the graph $ (\Phi\times \bbZ_+, E)$, where a directed edge exists from $(x,t)$ to $(y,t+1)$ if
$e_{xy}(t)=1$. Given $\Phi$, the location of the nodes is static, and the time evolution of the graph is entirely due to changes in the fading variables $h_t(u,v)$ and the set $\Phi_T$.

In this paper, we consider a space-time SINR graph in which the
location of the nodes is modeled as a homogeneous Poisson point
process (PPP). Modeling location of nodes in a wireless networks
as a PPP is quite attractive from an analytical point of view and
has paid rich dividends in terms of finding exact expressions for
several performance indicators such as maximum rate of
transmission (capacity), connection probability, etc.
\cite{Gupta2000, Weber2005, Baccelli2006, Andrews2011}, that are
hard to derive otherwise. A PPP node location model is well suited
for modeling both the ad hoc networks, where large number of
nodes are located in a large area without any
coordination, as well as the modern paradigm of cellular networks
\cite{Andrews2011}, where multiple different layers of
base-stations (BSs) (macro, femto, pico) are overlaid on top of
each other, and the union of all BSs appears to be uniformly distributed.

Given $\Phi$, the stochastic nature of the fading coefficients $h_t(\cdot,\cdot)$ and the set $\Phi_T(\cdot)$ implies that the event $e_{xy}(t)$ is a random variable, and hence potentially, multiple transmissions
are required for successfully transmitting a packet from node $x$
to $y$. 
Repeated transmissions entail {\it delay} in packet
transmission, and it is of interest to make the expected delay as small
as possible. Another related quantity of interest is the {\it
information velocity}, that is defined as the limit of the ratio
of the distance between the source and the destination of any
packet, to the total delay experienced by the packet to reach its
destination successfully over multi-hops, as distance goes to
infinity.

Expected delay and the information velocity are closely
connected to the various notions of capacities in wireless
networks, e.g.,  throughput capacity \cite{Gupta2000}, transport
capacity\cite{Gupta2000}, delay-normalized transmission capacity
\cite{Andrews2009, VazeTDR2011} etc., since all of them are measures based on the
successful rate of departure of packets towards their destination.
Finding the speed of information propagation is also related to first
passage percolation \cite{kesten1986aspects,hammersley1965first},
and dynamic epidemic processes
\cite{durrett1999stochastic,mollison1977spatial,mollison1978markovian},
however, the analysis in the space-time SINR graph gets 
complicated due to the presence of interference. 

In the seminal paper \cite{Gupta2000}, it was shown that with
PPP distributed node locations (albeit for a somewhat simpler model), 
the per-node throughput (rate of
transmission between any two randomly selected nodes) or
information velocity tends to zero as the size of the network grows. The most general analysis on the
expected delay and the information velocity has
been carried out in \cite{Baccellispacetime2011} for a PPP-driven space-time SINR
network. It is shown that with an ALOHA protocol, where nodes transmit independently of all other nodes with fixed power, the expected delay required for
a packet to leave a given node and be successfully received at any
other node in the network is infinite. Remarkably, this result is shown to hold 
even in the absence of interference and requires only the additive
noise to be present. Moreover, the information velocity is also shown
to be zero. These result have tremendous 'negative' impact on
network design, since it shows that essentially any packet cannot
exit its source with finite expected delay.

Both the results from \cite{Baccellispacetime2011}, viz., the infinite expected
delay and zero information velocity, are attributed to the fact
that with PPP distributed node locations, a typical node can have large
voids around itself, that is, regions that contain no other nodes with high probability. In such a circumstance, a large number of 
retransmission attempts will be required to overcome the
effect of additive noise and support a minimum
SNR at any of the other receiving nodes. Consequently, the
mean exit delay is infinite (when averaged over the realizations of the PPP) 
and the information velocity tends to zero.

One solution prescribed in \cite{Baccellispacetime2011} to make
the information velocity non-zero is to add another regular square
grid of nodes with a fixed density, in which case the nearest
neighbor distance is bounded, and  the information velocity can be
shown to be non-zero. From a practical point of view it is rather limiting to
assume the presence of such a regular grid. 

Some work has been reported on finite expected delay together with
a finite bound on the information velocity
\cite{jacquet2009opportunistic, ganti2009bounds, haenggibook},
under restrictive assumptions such as assuming temporal
independence with the SINR model, i.e. interference is independent
between any two nodes over time, no power constraint and temporal
interference independence, and no additive white Gaussian noise,
respectively.

In this paper, we propose a power control mechanism to show that
the information velocity can be made non-zero for the space-time
SINR graph with PPP node locations without any additional restrictive assumptions on the network. In \cite{Baccellispacetime2011}, the information velocity is defined as the limit of the ratio of the distance between two points $x$ and $y$ to the time it takes for the packet to go from $x$ to $y$ as the distance tends to infinity. The packet simultaneously traverses multiple paths and the time taken is the first time the packet is received at $y$. This makes the set-up somewhat complicated to work with and so the results in \cite{Baccellispacetime2011} are proved for delays averaged over the fading variables. In order to overcome this problem, we modify the definition of information velocity by specifying a random path along which a tagged packet will traverse the network. Allowing for a larger set of paths and picking the one that is optimal as done in \cite{Baccellispacetime2011} will only increase the velocity. Thus our result provides a lower bound on the information velocity. We describe our idea in brief here and the precise definition will appear later. 

We track a tagged packet as it traverses the network via a {\it conic} forwarding strategy along a prescribed path that depends only on the realization of the underlying distribution of the nodes. Briefly stated, conic forwarding works as follows. At each
node, the $\bbR^2$ plane is partitioned into multiple cones, and
each node transmits the packet at the head of its queue to its nearest neighbor in
the cone that contains the packet's destination until it is successfully received. We refer to such a cone as
the destination cone. This conic forwarding idea
circumvents the problem of forming nearest neighbor loops, since
the packet always progresses towards its destination. This also
allows us to exploit the various independences that exist across
space and time. The speed of this tagged particle along the prescribed path is what we will refer to as the information velocity (the direction of motion being contained in the choice of the transmission cones).  If $d(t)$ is the distance that the tagged packet travels in time $t$, then $v = \liminf_{t \to \infty} \f{d(t)}{t}$ is the information velocity. The aim of this paper is to devise a power-control strategy under which $v > 0$.

The power control strategy works by nullifying the path-loss from a node towards its nearest neighbor in the destination cone. In particular, since the path-loss between node $x$ and its nearest
neighbor $n(x)$ in the destination cone is $\ell(|x-n(x)|)$,
the transmitted power $P$ is taken to be $c \ell(|x-n(x)|)^{-1}$, where $c$ is a
constant. Since the nearest neighbor in a PPP can be at arbitrarily large distances, we need $\ell(\cdot) > 0$. To compensate for the non-homogeneity in power used at
each node, we modulate the transmission probability $p$, such that
$pP$ equals the average power constraint at each node.

We wish to note that per se, power control is not a new concept in wireless communication. For instance, in CDMA type wireless network \cite{Booktse}, power control is mandatory to overcome the near-far effect, where mobiles that are nearer to the base station have to continuously update their transmitted power so as not to severely limit the transmission from mobiles that are farther away. However, the specific strategy that we use in this paper has not been considered in the literature. Further the use and advantages of power control in large wireless networks with randomly located nodes has not been fully explored. 

In terms of implementation, the conic forwarding needs no special effort, since the transmitter only adjusts the power according to the nearest neighbor distance in the destination cone, and the transmission is isotropic, i.e., does not require any information about the direction, circumventing the need for any special hardware, e.g., directional antennas etc. For power control, the transmitter needs to learn the distance to its nearest neighbor in the destination cone. Nearest neighbor routing \cite{yu2001geographical, borbash2007asynchronous} is standard in wireless communication networks, where packets are forwarded to the nearest neighbors, which requires discovery of nearest neighbors as well as their distances, and hence our power control policy does not entail any new overhead. 

Using conic forwarding strategy together with power control,
we show that the expected delay to the nearest neighbor in the
destination cone is finite provided $\beta \gamma < 1$. In addition, as the tagged packet traverses the network from one node to another along the path specified by the conic forwarding strategy, the sequence of time delays turns out to be a non-stationary sequence. In order to overcome this problem,
we add additional (virtual) nodes to the network as the packet moves from one node to another. The nodes are added in such a way that the path along which the packet traverses in not affected. The interference experienced by the particle increases (and consequently reduces its speed) but the technique delivers for us a stationary sequence of delay times. This enables us to apply the ergodic theorem and infer with probability one that the information velocity is strictly positive for the stationary sequence, and hence for the original sequence. It is important to note that these virtual nodes are not really required in practice to achieve non-zero information velocity, but are only used as a analytical tool to upper bound the per-hop delays (via increasing the interference) and making them stationary across different hops. 
\section{System Model}
Let $\Phi$ be a homogenous PPP with intensity $\lambda$ in
$\bbR^2$ modeling the location of the nodes of the network. The
time parameter is assumed to be discrete (slotted). Let
$\{h_t(x,y), x,y \in \Phi, t=0,1, \ldots \}$ be a collection of
independent exponentially distributed random variables with
parameter $\mu$. $h_t(x,y)$ is the fading power from
node $x$ to node $y$ in the time slot $t$. The path loss between
$x,y \in \Phi$ denoted by $\ell(x,y) = \ell(|x-y|)$ is given by
\[ \ell(r) = r^{-\al} \wedge 1, \qquad r > 0, \]
where $a \wedge b = \min(a,b)$ and $\al > 2$ is arbitrary.

We assume that each node can only operate in a half-duplex mode,
that is, in the time slot $t$, a node $x\in \Phi$ is {\it on}
(transmitter) or {\it off} (receiver) following a Bernoulli random
variable $\1_x(t)$, with $\bbP(\1_x(t) = 1) = p_x(t)$. Let
$q_x(t)=1-p_x(t)$. The set of {\it on} ({\it off}) nodes in the
time slot $t$ is denoted by $\Phi_T(t)$ ($\Phi_R(t)$).

\begin{figure}
\centering
\includegraphics[height=2.5in]{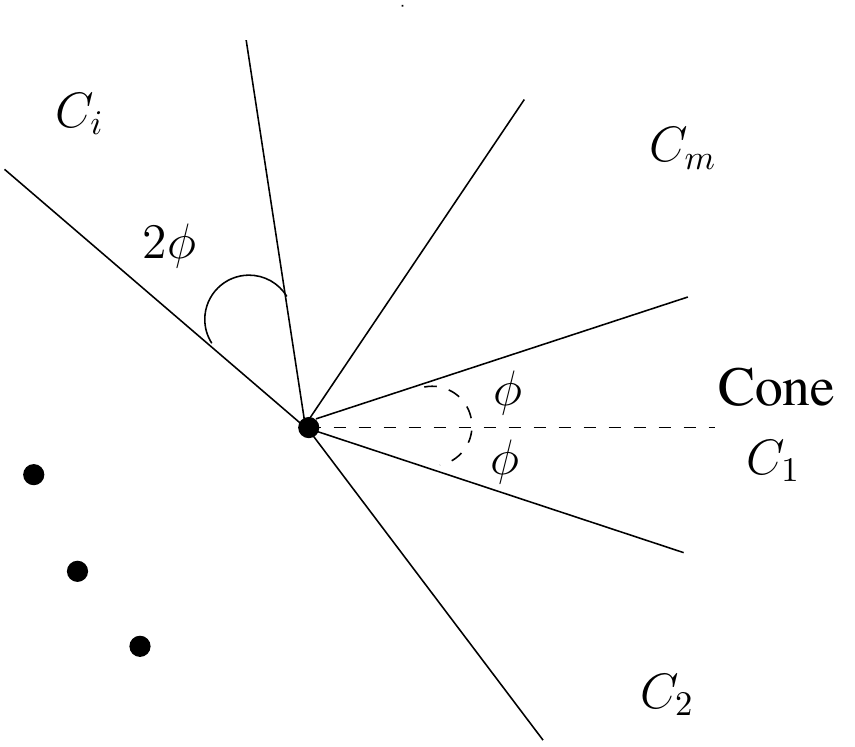}
\caption{Definition of cones with angle $2\phi$.}
\label{fig:conedef}
\end{figure}

\begin{figure}
\centering
\includegraphics[height=2.5in]{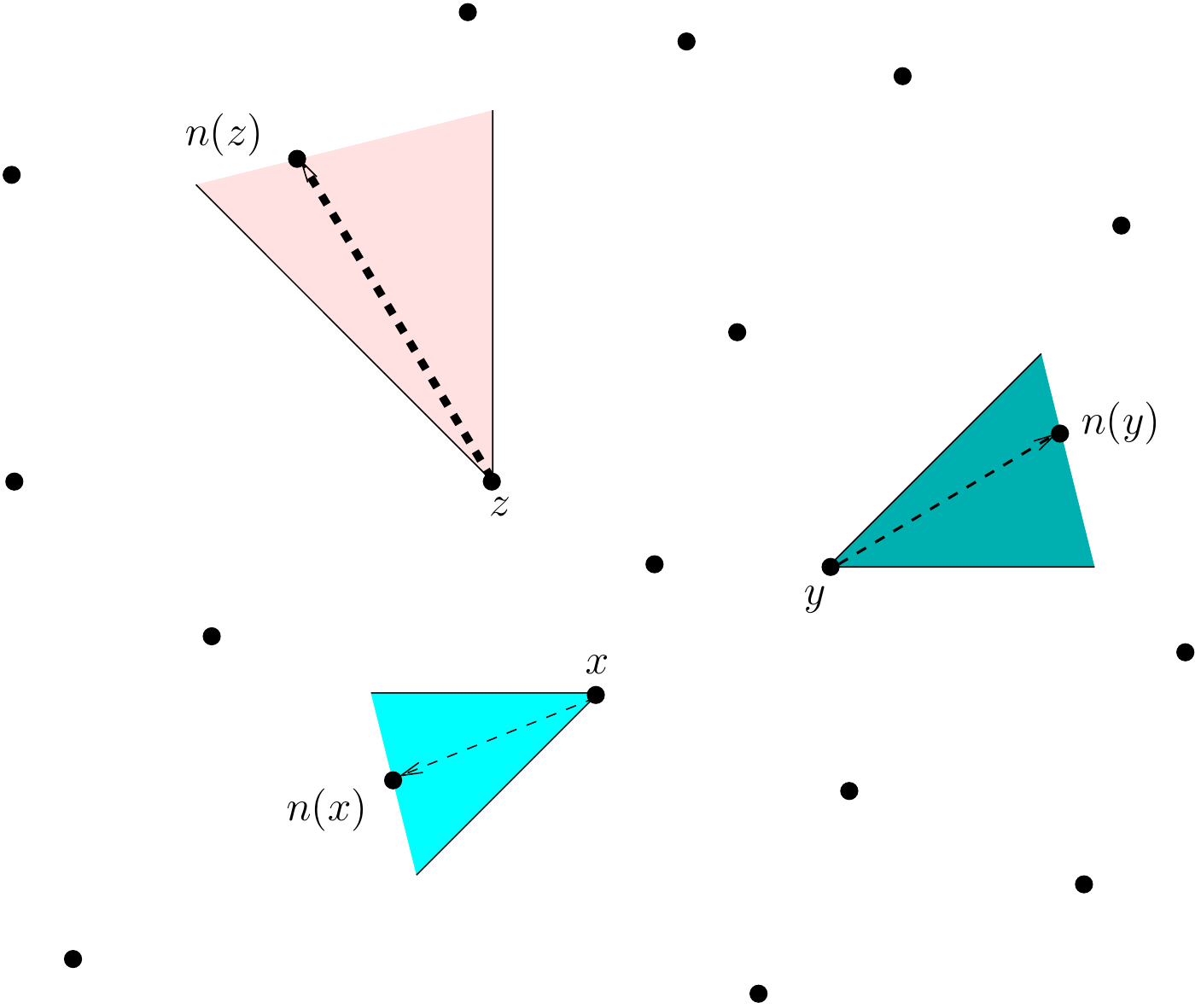}
\caption{Each node transmits to its nearest neighbor in the destination (shaded) cone.}
\label{fig:conetransmit}
\end{figure}

Let $C_1,\ldots, C_m$ be cones with angle $2\phi < \f{\pi}{2}$ in
$\bbR^2$ with vertex at the origin, satisfying $\cup_{i=1}^m C_i =
\bbR^2$ and $C_i \cap C_j = \emptyset$ for $i \ne j$, as shown in
Fig. \ref{fig:conedef}. Without loss of generality, suppose that
$C_1$ is symmetric about the x-axis and opens to the right.  Let
$x+C_i$ be translation of cone $C_i$ by $x$. In the time slot $t$,
for a node $x$, let $x+C_d(x,t)$ be the cone that contains the
final destination of the packet that it wishes to transmit. We call this cone as the
destination cone. Denote the nearest neighbor of $x$ in the
destination cone $x+C_d(x,t)$ by $n_t(x)$. 
If the node at $x$ is {\it on} in time slot $t$ then it transmits with power $P_x(t) $. 
The key idea in this paper is to employ a decentralized power control scheme, that is, the functions $p_x(t), P_x(t)$ depend locally on $\Phi$. The particular forms that these functions take are given by
\begin{equation}
P_x(t) = c \ell^{-1}(x,n_t(x)), \qquad p_x(t) = M (P_x(t))^{-1},
\label{eq:power_prob}
\end{equation}
where $c = M (1-\ep)^{-1}$, $0< \epsilon <1$ is a constant, and $M = P_x(t)p_x(t)$ is the average
power constraint. Note that $p_x(t) \leq 1 - \ep$, since $\ell(\cdot) \leq 1$. Thus, in each time slot, each node makes transmission attempts with transmission power proportional to the
distance to its nearest neighbor in the destination cone to
compensate for the path-loss to the nearest neighbor. The transmission
probability is chosen so as to satisfy an average power constraint.

In Fig. \ref{fig:conetransmit}, we illustrate the transmission
strategy, where each node transmits to its nearest neighbor in the
destination cone (shaded cone) with line thickness proportional to
the transmit power, farther the nearest neighbor, larger the
power. In prior work \cite{Baccellispacetime2011}, with the ALOHA protocol, the functions
$P$ and $p$ were assumed to be constants that were independent of the system parameters.

Thus, the SINR from node $x$ to node $y$
in time slot $t$ is given by
\begin{equation}\label{eq:SINR}
\SINR_{xy}(t) = \frac{P_x(t) h_t(x,y) \ell(x,y) \1_x(t) (1 -
\1_{y}(t))}{ \ga \sum_{z \in \Phi_T(t) \backslash \{x,y\}} P_z(t)
h_t(z,y) \ell(z,y) + \sfN},
\end{equation}
where $0 < \gamma < 1$ is the processing gain of the system (interference
suppression parameter) which depends on the transmission/
detection strategy, for example, on the orthogonality between
codes used by different legitimate nodes during simultaneous
transmissions.
The transmission from node $x$ to
$y$ is deemed successful at time $t$, if $\SINR_{xy}(t) > \beta$,
where $\beta > 0$ is a fixed threshold. Let $e_{xy}(t) = 1$ if
$\SINR_{xy}(t) > \beta$, and zero otherwise. The sum in the
denominator of the right hand expression in the above equation is
referred to as the interference and $\al > 2$ ensures that the
interference term in the denominator is finite almost surely.
Since $h_t(x,y)$ is exponentially distributed, multiple
transmissions may be required for a packet to be successfully
received at any node.

\section{Main Results and Proofs}
Our first objective is to show that with the power control
policy described above, the expected time for a packet to be successfully received
at the nearest neighbor in the destination cone is finite.

\begin{defn} Let the minimum time (exit time) taken by any
packet to be successfully transmitted from node $x$ to its nearest neighbor $n(x)$
in the destination cone of the packet be
$$T(x) = \min\left\{ t >0 : e_{x, n_t(x)}(t) =1  \right\}.$$
\end{defn}
Let $\bbP^x$ denote the Palm distribution of $\Phi$, conditioned to have a point at $x$ and let $\bbE^x$ denote expectation with respect to $\bbP^x$. By an abuse of notation we will use $\bbP$ and $\bbE$ to denote $\bbP^o$ and $\bbE^o$. This will cause no confusion since these are the only probabilities and expectations that are of interest.
\begin{thm}\label{thm:finiteexittime} [Finite expected exit time] Suppose $\beta \ga < 1$.
Then for all $\ep > 0$ sufficiently small, the space-time SINR
graph with power control policy as described above satisfies
$\bbE^x\{T(x)\} < \infty$.
\end{thm}
\begin{remark}  The parameter $\beta$ controls per link data rate, larger the value of $\beta$, larger is the per link data rate. The condition $\beta \gamma < 1$ indicates that to support larger per-link data rate, one has to invest in getting a better (lower) interference suppression parameter, e.g. by lowering the chip rate in a wireless CDMA system. The condition $\beta \gamma < 1$ also indicates that there is no free lunch, i.e., if one wants larger data rate and finite expected exit time, one has to have better interference suppression capability. An interesting upshot of the proposed power control policy is that the condition required for the theorem to hold is independent of the intensity $\lambda$ of the PPP.
\end{remark}
\begin{proof}
Without loss of generality, suppose that the origin $o \in \Phi$, that is we will consider the point process under $\bbP^o$ which, as we noted above will be denoted by $\bbP$.
We tag a particular packet to be transmitted out of the node at
$o$, and suppose that the destination cone of this packet when it is at $o$ is $C_d$.
Denote the nearest neighbor of $o$ in $C_d$ by $n(o)$. We have dropped the time subscript from $n(o)$, since, as long as the packet is not successfully transmitted out of $o$ the destination cone remains the same. Let

\begin{equation}\label{eq:newSINR}
\SINR_{o,n(o)}(t) = \frac{P_o(t) h_t(o,n(o)) \ell(o,n(o)) \1_o(t)
(1 - \1_{n(o)}(t))}{\ga I(t)  + \sfN},
\end{equation}
where
\begin{equation}\label{eq:redintf}
I(t) = \sum_{z \in \Phi \backslash \{o, n(o) \}} \1_z(t)
P_z(t) h_t(z,n(o)) \ell(z,n(o)).
\end{equation}
Let $e_{o,n(o)}(t) =1$ if $\SINR_{o,n(o)}(t) > \beta$, and $0$
otherwise. Due to interference and the nature of the traffic arriving at the nodes, the choice of the destination cones are not independent across time slots at the same node as well as at different nodes. Hence to evaluate $\bbE\{T(o)\}$ we need to condition appropriately. Let
$\cG_k$ be the sigma field generated by the point process $\Phi$
and the choice of cones made at all nodes of $\Phi$ at times
$t=1,2, \ldots ,k$. Note that as long as the packet is not
successfully transmitted out of $o$, the transmission probability
$p_o(t)$ does not change. Let $F := \cap_{j=2}^k \{ p_o(j) =
p_o(1)\}$. Now
\begin{equation}\label{eq:dummy1}
 \bbP\left[T(o) > k \big| \Phi \right] =  \bbE \left\{\bbP\left[e^d_{o,n(o)}(t) =0, \; \forall \;
t=1,\dots, k \big| \cG_k \right] \1_F \Big| \Phi \right\}.
\end{equation}
Let $A(t)$ be the event that the origin $o \in \Phi_R(t)$, and
$B(t)$ be the event that $o \in \Phi_T(t)$, $n(o) \in \Phi_R(t)$
but $\SINR_{o,n(o)}(t) \le \beta$. Writing right hand side of \eqref{eq:dummy1}
 in terms of $A(t), B(t)$, and using the independence of
the fading powers and the conditional independence of the
transmission events, we get
\begin{equation}
\label{eq:indep} \bbP\left[T(o) > k \big| \Phi \right]  =  \bbE \left\{
\prod_{t=1}^k \bbP \left[ A(t) \cup B(t) \big| \cG_k \right] \1_F \Big| \Phi \right\}.
\end{equation}
On the event $F$, we have $\bbP(A(t)|\cG_k) = 1 - p_{o}(1)$ and
\begin{equation}
\bbP(B(t)|\cG_k) \label{eq:E2}
 = p_{o}(1) q_{n(o)}(t)
\left(1-\bbE\left\{\exp\left(- \f{\mu \beta}{c} (\sfN+ \ga
I(t))\right)\bigg{|}\cG_k\right\}\right),
\end{equation}
for $t=1,2, \ldots ,k$. \eqref{eq:E2} follows from
\eqref{eq:newSINR} by using the fact that $P_o(t) \ell(o,n(o)) =
c$ and taking expectation with respect to $h_t(o,n(o)) \sim
\exp(\mu)$. This yields
\begin{eqnarray} \nn
\bbP \left[ A(t) \cup B(t) \big| \cG_k \right] & \leq & 1- p_{o}(1) + p_{o}(1)
q_{n(o)}(t) \left(1-\bbE\left\{e^{- \f{\mu \beta}{c} (\sfN+ \ga I(t))}\Big{|}\cG_k \right\}\right)\\
\nn & = & 1 - p_{o}(1)p_{n(o)}(1)  - p_{o}(1) q_{n(o)}(t) e^{- \f{\mu \beta \sfN}{c}
} \bbE\left\{ e^{- \f{\mu \beta \ga}{c}I(t)} \Big|\cG_k \right\} \\
\label{eq:E3} & \leq & 1 - p_{o}(1) \, \ep \, e^{- \f{\mu \beta \sfN}{c}
}\, \bbE\left\{ e^{- \f{\mu \beta \ga}{c}I(t)} |\cG_k \right\},
\end{eqnarray}
where we have used the fact that $q_y(t) \geq \ep$. Let $a =
\f{\mu \beta \ga}{c}$. We now find a lower bound for the
expectation on the right hand side above that is independent of
the choice of the cone. To this end observe that
\begin{equation} \label{eq:mgfint}
\bbE\left\{ e^{- \f{\mu \beta}{c}I(t)} \Big|\cG_k \right\} = \prod_{z
\in \Phi \backslash \{o, n(o) \}} \bbE\left\{ e^{- a \1_z(t)
P_z^{(i)}(t) h_t(z,n(o)) \ell(z,n(o))} \Big|\cG_k \right\}
\end{equation}
Suppose node $z \in \Phi \backslash \{o, n(o) \}$ transmits using
cone $z+C_i$ in time slot $t$. This fixes the transmission
probability $p_z^{(i)}(t)$ and power $P_z^{(i)}(t)$ (where we have included the index $i$ to make the dependence on the cone explicit). Then
\begin{eqnarray} \nn \lefteqn{ \bbE\left\{e^{- \f{\mu \beta \ga}{c} \1_z(t) P_z^{(i)}(t)
h_t(z,n(o)) \ell(z,n(o))} \Big| \cG_k \right\}} \\
\nn & = & (1-p_z^{(i)}(t)) +  p_z^{(i)}(t) \bbE\left\{ e^{ -
a P_z^{(i)}(t) h_t(z,n(o)) \ell(z,n(o))} \Big| \Phi \right\} \\
\label{eq:bound_indiv_int} & = & (1-p_z^{(i)}(t)) + p_z^{(i)}(t) \f{c}{c
+ \beta \ga \ell(z,n(o)) P_z^{(i)}(t)}.
\end{eqnarray}
Let $C_z^* = C_z^*(\Phi)$ be the cone for which the right hand
expression in \eqref{eq:bound_indiv_int} is minimized, i.e., node $z$ causes maximum interference at $n(o)$ when its destination cone choice is $C_z^*$. Let
$p_z^*$, $P_z^*$ denote the corresponding transmission probability
and power respectively. Denote by $\1^*_z$ an independent Bernoulli random
variable with $\bbP[\1^*_z = 1] =p_z^*$. The cone $C_z^*$ maximizes the interference contribution at $n(o)$ due to transmission at $z$. Define
\begin{equation}
I^*(t) = I^*(t,\Phi) = \sum_{z \in \Phi \backslash \{o, n(o) \}}
\1^*_z P_z^* h_t(z,n(o)) \ell(z,n(o)). \label{eq:interference_max}
\end{equation}
Substituting $I^*(t)$ for $I$ in \eqref{eq:E3} along with the observation that
given $\Phi$, $I^*(t) \stackrel{d}{=} I^*(1)$ we get
\begin{equation}
\bbP \left[ A(t) \cup B(t) \big|\cG_k \right] \leq 1 - p_{o}(1) \, \ep \, e^{- \f{\mu \beta
\sfN}{c}} \, \bbE\left\{ e^{- a I^*(1)} \Big| \Phi \right\}.
\label{eq:E4}
\end{equation}
Substituting from \eqref{eq:E4} in \eqref{eq:indep}, we get
\begin{equation}\label{eq:J}
\bbP\left[T(o) >  k \big| \Phi \right]  \leq  (1-J)^k,
\end{equation}
where
\[ J = p_{o}(1) \ep e^{- \f{\mu \beta \sfN}{c}} \bbE\left\{\exp\left(-a I^*(1) \right) \big| \Phi\right\}. \]
The expected delay can then be written as
\begin{eqnarray*}\nn
\bbE\{T(o)\}& = & \sum_{k\ge 0} \bbP[T(o) >  k] \\
\nn & = & \bbE\left\{\sum_{k\ge 0} \bbP\left[T(o) > k \big| \Phi \right]\right\} \\
\nn & \leq & \bbE\{J^{-1}\}. 
\end{eqnarray*}
By the Cauchy-Schwartz inequality we get
\begin{equation}
\label{eq:CS} \bbE\{T(o)\} \le  \frac{e^{\f{\mu \beta
\sfN}{c}}}{\epsilon} \left(\bbE\left\{\f{1}{ \left( \bbE\left\{ e^{-a
I^*(1)} \big| \Phi \right\}\right)^2 } \right\} \;
\bbE\left\{p_{o}(1)^{-2}\right\} \right)^{\f{1}{2}}.
\end{equation}
From the definition of the transmission probability $p_{o}(t)$, we
get
\begin{equation}\label{eq:p(o)bound}
 \bbE[p_o(1)^{-2}] \leq \bbE\left\{  \left(\f{c}{M} \right)^{2} (|n(o)|^{2 \al} \vee 1) \right\} < \infty,
\end{equation}
since the nearest neighbor distance in a cone has density
\begin{equation}
f(r) = \f{2\lambda \pi r}{m} e^{- \f{\lambda \pi}{m}r^2}, \qquad r
> 0. \label{pdfnnd}
\end{equation}
It remains to show that
\begin{equation} \bbE\left\{\f{1}{ \left( \bbE\left\{ e^{-a I^*(1)} \big| \Phi \right\} \right)^2
} \right\} < \infty. \label{eqn:finite_mean_int}
\end{equation}
\begin{equation}
\bbE\left\{ e^{-a I^*(1)} \big| \Phi \right\}  = \prod_{z \in \Phi
\backslash \{o, n(o) \}} \bbE \left\{ e^{- a \1^*_z P_z^*
h_1(z,n(o)) \ell(z,n(o))} \big| \Phi  \right\}.
\label{eqn:cond_exp_int}
\end{equation}
Taking expectations, first with respect to $\1^*_z$ and then with respect to $h_1(z,n(o))$, we get
\begin{eqnarray*}
\bbE \left\{ e^{- a \1^*_z P_z^* h_1(z,n(o)) \ell(z,n(o))} | \Phi
\right\} & = & (1-p^*_z) + p^*_z \bbE \left\{ e^{- a
P_z^* h_1(z,n(o)) \ell(z,n(o))} \big| \Phi \right\} \\
 & = & 1 - p_z^* \left( 1 - \f{\mu}{\mu + a P_z^* \ell(z,n(o))}
\right) \\
 & \stackrel{(a)}= & 1 - \f{\beta \ga p_z^* P_z^* \ell(z,n(o))}{c + \beta \ga P_z^*
\ell(z,n(o))} \\
 & \stackrel{(b)}\geq & 1 - \f{\beta \ga M \ell(z,n(o))}{c} = 1 - \beta \ga (1-\ep)
 \ell(z,n(o)),
\end{eqnarray*}
where $(a)$ follows by substituting $\f{\mu \beta \gamma}{c}$ for $a$, and to obtain $(b)$ we have used the fact that the average power $p_z P_z$ equals $M$ for $\forall \ z$ and in 
particular $p^*_z P^*_z = M$ and $c = M(1-\ep)^{-1}$. Let $c_1=\beta \ga (1-\ep)$. Substituting the above bound in (\ref{eqn:cond_exp_int}) we get
\begin{equation}
\bbE\left\{ e^{-a I^*(1)} \big| \Phi \right\} \geq \prod_{z \in \Phi
\backslash \{o, n(o) \}} \left( 1 - c_1 \ell(z,n(o)) \right).
\label{eqn:bound_cond_exp_int}
\end{equation}
Note that $c_1 \ell(z,n(o)) < 1$ since $\beta \gamma < 1$. Let $B(x,r)$ denote a ball of radius $r$ centered at $x$. Substituting the bound obtained in (\ref{eqn:bound_cond_exp_int}) in (\ref{eqn:finite_mean_int}),
we get
\begin{equation}
\bbE\left\{\f{1}{ \left( \bbE\left\{ e^{-a I^*(1)}| \Phi \right\}
\right)^2 } \right\}  \leq  \bbE\left\{ \prod_{z \in (\Phi
\backslash \{o, n(o) \})  \cup \Phi_0} e^{ - 2 \log \left( 1 - c_1
\ell(|z|) \right)} \right\}, \label{eq:expintfbound1}
\end{equation}
where the last inequality follows by shifting the origin to $n(o)$ and including points from an independent Poisson process $\Phi_0$ of intensity $\lambda \1_{\{(o+C_d) \cap B(o,|n(o)|)\}}$, i.e., a PPP of intensity $\lambda$ restricted to the set $(o+C_d) \cap B(o,|n(o)|)$. 
Clearly $(\Phi \backslash \{o, n(o) \}) \cup \Phi_0$ is a PPP of intensity $\lambda$ with the origin at $n(o)$. Hence by an application of the Campbell's theorem in (\ref{eq:expintfbound1}) and the fact that
$\ell(|z|) \leq 1$ we get
\begin{eqnarray}
\bbE\left\{\f{1}{ \left( \bbE\left\{ e^{-a I^*(1)}| \Phi \right\}
\right)^2 } \right\}  & \leq & \exp\left( \lambda \int_{\bbR^2}
\left( e^{-2 \log \left( 1 - c_1 \ell(|z|) \right)}
- 1 \right) {\mathrm d}z \right) \nn \\
 & \leq & \exp \left( \f{2 \lambda c_1}{(1-c_1)^2} \int_{\bbR^2} \ell(|z|)
{\mathrm d}z \right) < \infty, \label{eqn:finite_exp_int}
\end{eqnarray}
since $\al > 2$. This completes the proof of Theorem~\ref{thm:finiteexittime}.
\end{proof}

Next, we build upon Theorem \ref{thm:finiteexittime}, to show that
the information velocity, that is, the rate at which packets flow towards
their destination, is positive under the proposed power control
mechanism. Information velocity is a key quantity in multi-hop
routing. Larger the velocity, higher is the capacity of the network.
The negative results in \cite{Baccellispacetime2011} on the infinite expected delay and zero information velocity are proved for delays that are averaged over the fading variables. In order to work with the delay variables directly and also to be able to use the ergodic theorem, we introduce several additional structures as we go along. 

As a first step we track the movement of a tagged packet that starts
at the origin $X_0 = o \in \Phi$ and traverses the network as
follows. Let $T_0$ be the time taken by this tagged packet starting at the
origin to successfully reach its nearest neighbor $X_1 = n(o)$ in
the destination cone $C_1$. The packet is transmitted with power $P_1 = c \ell(|n(o))|)^{-1}$
and the probability that it is transmitted in a time slot is $MP_1^{-1}$. Going
forward, if the packet is at node $X_{i-1}$, $i \geq 2$, let $T_{i-1}$ be the
time taken for the packet to successfully reach the nearest
neighbor $X_{i}$ of $X_{i-1}$ in the destination cone $X_{i-1} + C_1$. From the time the packet reaches $X_{i-1}$ to the time it is successfully delivered at $X_i$, it is transmitted with
power $P_i = c \ell(|X_i - X_{i-1})|)^{-1}$ and transmission probability $MP_i^{-1}$. One can think of the
destination of the packet being located at $(\infty,0)$ and thus the destination cone for this packet is always a translation of the cone $C_1$. We ignore the queuing and other delays at each node as is the standard practice. Note that the delays $T_i,$ $ i \geq 0,$ are not identically distributed. For instance, the point process as seen from the origin and from $n(o)$ do not have the same distribution, since the area of �destination cone between $o$ and $n(o)$ contains no other point of $\Phi$. In particular, $\{T_i, i \geq 0\}$ is not a stationary sequence.

\begin{defn}
The information velocity of space-time SINR network is defined as
\[ v  = \liminf_{t \to \infty} \f{d(t)}{t},\]
where $d(t)$ is the distance of the tagged packet from the origin at time $t$.
\end{defn}

The following is the main result of this paper.

\begin{thm} \label{thm:positiveinfovelocity}
Under the conditions of Theorem~\ref{thm:finiteexittime} and the
proposed power control strategy, the information velocity $v
> 0$, almost surely.
\end{thm}

\begin{proof}
In order to prove this result, we first dominate the delays
$\{T_i, i \geq 0\}$, by a stationary sequence $\{T_i^{'}, i \geq
0\}$, and show that a positive speed can be obtained even with
these enhanced delays. This will be done by adding some additional points that will increase the interference and hence the delay. To this end, for all $i \geq 0$, let $R_i =
|X_{i+1} - X_i|$ and $\th_i = \arcsin ((X_{i+1,2}-X_{i,2})/R_i)$,
where $X_i = (X_{i1},X_{i2})$. Note that the cones 
$\{(X_i+C_1) \cap B(X_i, R_i), i \geq 0\}$ are non-overlapping since $\phi < \pi/4$.
Consequently, $\{(R_i,\th_i), i \geq 0\}$ is a sequence of independent and
identically distributed random vectors having the same
distribution as the random vector $(R, \th)$, where $R$ and $\th$
are independent with density of $R$ given by (\ref{pdfnnd}) and
$\th$ is uniformly distributed on $(-\phi, \phi).$

To nullify the effect of moving to the nearest neighbor, we progressively fill the voids with independent Poisson points as the packet traverses the network. This however leaves an increasing sequence of special points, the $X_i$\rq{}s in the wake of the tagged packet. The following construction is intended to take care of this issue and deliver a stationary sequence.

\begin{figure}
\centering
\includegraphics[height=2.5in]{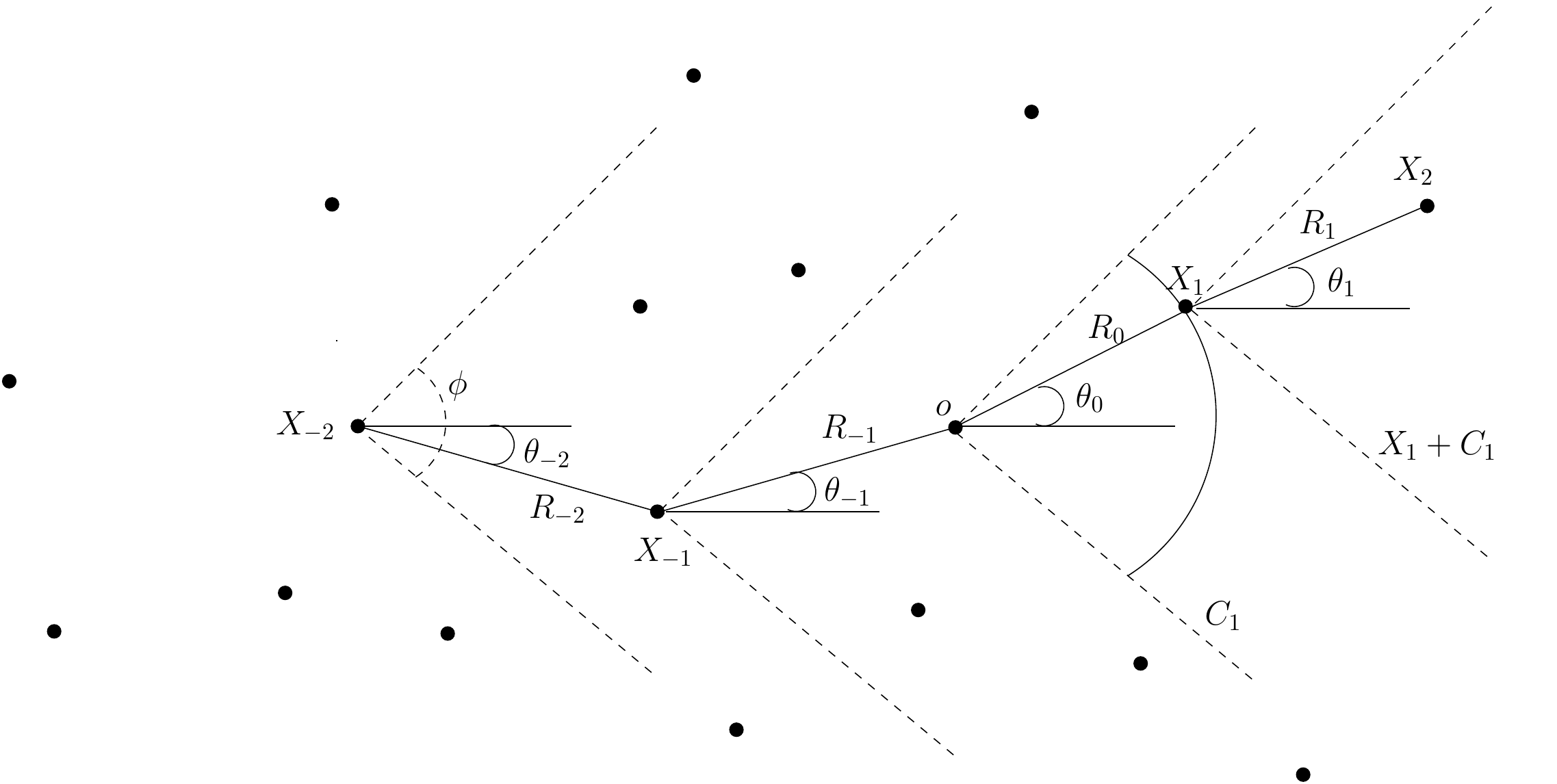}
\caption{Illustration of addition on infinite sequence of points/interferers to make $T_i'$ stationary.} \label{fig:biinfext}
\end{figure}

Let $\{(R_{-i}, \th_{-i}, i \geq 1\}$ be a sequence of independent
random vectors with each vector having the same distribution as
$(R,\th)$. Define $\tilde{\Phi} = \{X_{-i}, i \geq 1\}$,
recursively starting from $X_{-1}$ so as to satisfy $|X_{-i} -
X_{-i+1}| = R_{-i}$ and $\th_{-i} = \arcsin ((X_{-i+1,2}-X_{-i,2})/R_{-i})$. Observe that
each $X_{-i+1},$ $i \geq 1$,  lies in the cone $X_{-i} + C_1$ as shown in Fig. \ref{fig:biinfext}, and $\{\Phi \cap ((X_{-i}+C_1) \cap B(X_{-i},R_{-i})), i \geq 1\}$ is a sequence of independent and identically distributed random variables. 

Let $T_0^{'}$ be the delay experienced by
the tagged packet in going from $X_0$ to $X_1$ when the
interference is coming from the points of $(\Phi\setminus \{o, n(o)\})
\cup \tilde{\Phi}$. For $i \geq 0$, let $\Phi_i$ be an PPP of intensity $\lambda \1_{\{(X_{i}+C_1)\cap B(X_i,R_i)\}}$ independent of everything else. 
For $i \geq 1$, let $T_i^{'}$ be the delay experienced by the
tagged particle in going from $X_i$ to $X_{i+1}$ when the
interference is coming from the nodes in $ (\Phi
\setminus \{X_i, X_{i+1}\}) \cup \tilde{\Phi} \cup_{j=0}^{i-1} \Phi_j$. Note that for the actual delay $T_i$, that is, when the packet is at $X_i$ and trying to reach $X_{i+1}$, the interference contribution is coming from the nodes in $\Phi \setminus \{X_i,X_{i+1}\}$. To define $T_i^{'}$, we have added additional interferers at $\tilde{\Phi} \cup_{j=0}^{i-1} \Phi_j$. We assume that virtual interferers placed at $\tilde{\Phi} \cup_{j=0}^{i-1} \Phi_j$ behave similar to nodes of $\Phi$. Clearly $T_i^{'} \geq
T_i$, and furthermore the sequence $\{T_i^{'}, i \geq 0 \}$ is a stationary sequence. To prove the later assertion, consider any finite dimensional vector of delays $(T_ {i_1}^{\rq{}}, T_{i_2}^{\rq{}}, \ldots ,T_{i_j}^{\rq{}})$. The distribution of this vector is a function of the distribution of the special points $\{X_i, i \leq i_j\}$, the point processes $(\Phi \setminus \cup_{i \leq i_j} X_i)$ and $\cup_{i = 0}^{i_j -1} \Phi_i$, which by our construction is translation invariant.

Suppose we showed that $\eta = \bbE[T_i^{'}] < \infty$. Then by the
Birkoff's ergodic theorem \cite{ergodictheorybook}, we have
\[ \lim_{n \to \infty} \f{1}{n} \sum_{k=0}^{n-1} T_k^{'} = T^{'},\]
almost surely, where $T^{'}$ is a random variable with mean
$\eta$.

Let $N(t)$ be the counting process associated with an arrival
process with inter-arrival times given by the sequence $\{T_i^{'},
i \geq 0\}$. Then the information velocity satisfies
\[ v \geq \lim_{t \to \infty} \f{\sum_{k=1}^{N(t)} R_k \cos(\th_k)}{\sum_{k=1}^{N(t)+1}
T_k^{'}} = \f{\bbE[R\cos(\th)]}{T^{'}}. \]
The result now follows since $T^{'}$ is finite almost surely. 

It remains to show that $E[T_0^{'}] < \infty$.
The proof of this assertion proceeds along the same lines as the proof of
Theorem~\ref{thm:finiteexittime} with $I(t)$ in (\ref{eq:redintf})
replaced by $I(t) + \tilde{I}(t)$, where
\[ \tilde{I}(t) = \sum_{z \in \tilde{\Phi}} 1_z(t) P_z(t)
h_t(z,n(o)) \ell(z,n(o)). \]
This would lead to a bound analogous to (\ref{eq:CS}) with
$I^*(1)$ replaced by $I^*(1) + \tilde{I}^*(1)$ and $\Phi$ replaced
by $\Phi \cup \tilde{\Phi}$, where $\tilde{I}^*(1)$ is defined
analogous to $I(1)$. By the conditional independence of $I^*(1)$
and $\tilde{I}^*(1)$ we get
\[ \bbE\left\{\f{1}{ \left( \bbE\left\{ e^{-a (I^*(1) +
\tilde{I}^*(1))}|\Phi \cup \tilde{\Phi}\right\} \right)^2 }
\right\} = \bbE\left\{\f{1}{ \left( \bbE\left\{ e^{-a I(1)}| \Phi
\right\} \bbE\left\{ e^{-a \tilde{I}(1)}| \tilde{\Phi} \cup \{n(o)
\} \right\} \right)^2 } \right\} \]
Another application of the Cauchy-Schwartz inequality implies that
the result follows if we show that
\[ \bbE\left\{\f{1}{ \left(
\bbE\left\{ e^{-a I(1)} \big| \Phi \right\} \right)^4} \right\}
\bbE\left\{ \f{1}{ \left( \bbE\left\{ e^{-a \tilde{I}(1)} \big|
\tilde{\Phi} \cup \{n(o) \} \right\} \right)^4 } \right\} <
\infty. \]
Proceeding as in
(\ref{eq:expintfbound1})-(\ref{eqn:finite_exp_int}), we get
\[ \bbE\left\{\f{1}{ \left( \bbE\left\{ e^{-a
I(1)} \big| \Phi \right\} \right)^4} \right\} \leq \exp \left( \f{
\lambda}{(1-c_1)^4} \int_{\bbR^2} \left( 1 - (1-c_1 \ell(|z|))^4
\right) {\mathrm d}z \right) < \infty, \]
since $\al > 2.$ It remains to show that
\begin{equation} \bbE\left\{\f{1}{ \left( \bbE\left\{ e^{-a \tilde{I}(1)} \big|
\tilde{\Phi} \cup \{n(o) \} \right\} \right)^4 } \right\} <
\infty. \label{eqn:finite_mean_add_int}
\end{equation}
To compute the expression in (\ref{eqn:finite_mean_add_int}) we proceed as we did in
(\ref{eqn:cond_exp_int})-(\ref{eqn:bound_cond_exp_int}) and arrive at the following bound similar to the one in (\ref{eq:expintfbound1}).
\begin{eqnarray} \nn
\bbE\left\{\f{1}{ \left( \bbE\left\{ e^{-a \tilde{I}(1)} \big|
\tilde{\Phi} \cup \{n(o) \} \right\} \right)^4 } \right\}  & \leq
& \bbE\left\{ \prod_{i=1}^{\infty} e^{-4 \log \left( 1 - c_1
\ell(X_{-i}, n(o)) \right)} \right\}  \\
 \nn & \leq & \bbE\left\{ \prod_{i=1}^{\infty} e^{-4 \log \left( 1 - c_1
\ell(\sum_{j=0}^i R_{-j} \cos(\th_{-j}))\right) } \right\} \\
\label{eq:expboundaddlintf} & \leq & \bbE\left\{ e^{
\sum_{n=1}^{\infty} g(S_{n+1})} \right\},
\end{eqnarray}
where $S_n = \sum_{j=0}^{n-1} R_{-j} \cos(\th_{-j})$ and $g(x) =
-4\log(1-c_1 \ell(x))$, $x > 0$. Let $\xi = \bbE\{R \cos(\th)\}$
and note that $\xi
> 0$. Let $\del \in (0, \xi)$ be a constant that will be chosen
later. Define $\chi(\nu) = \bbE\{ e^{\nu R \cos(\th)} \}$, $\nu \in \bbR$ and let $\zeta(\del) = \inf\{ \nu \del - \log (\chi(\nu)), \nu > 0 \}.$ That $\chi(\nu) < \infty$ for all $\nu$ follows from (\ref{pdfnnd}). By the Chernoff bound, we have
\[ \bbP\left[ \f{S_n}{n} < \del \right] \leq e^{-\zeta(\del) n}. \]
It follows by the Borel-Cantelli lemma that, almost surely, there exists a $N=N(\om) < \infty$ such that $S_n \geq n \del$ for all $n \geq N(\om)$. Hence for some constant $c_2 > 0$,
\begin{eqnarray} \nn
P[ N \geq m] & = & \bbP[ S_n < n \del \mbox{ for some } n \geq m],
\\ \label{eq:chernoffbound}
 & \leq & \sum_{n=m}^{\infty} e^{- \zeta(\del) n} \leq c_2 e^{- \zeta(\del)
 m}.
\end{eqnarray}
Using the fact that the function $g$ is non-increasing, we get
\begin{eqnarray*}
\bbE\left\{ e^{\sum_{n=1}^{\infty} g(S_n)} \right\} & = & \bbE\left\{
e^{\sum_{n=1}^{N} g(S_n) + \sum_{n=N +
1}^{\infty} g(S_n)} \right\}, \\
 & \leq & e^{\sum_{n=1}^{\infty} g(n\del)} \bbE\left\{ e^{g(0)
 N} \right\}.
\end{eqnarray*}
$\sum_{n=1}^{\infty} g(n\del) < \infty$ since $\al > 2$ by the
comparison test. Since $R \cos(\th) > 0$, $\zeta(\del) \uparrow
\infty$ as $\del \downarrow 0$. So, we can and do choose
$\del$ such that $\zeta(\del)
> g(0)$. With this choice of $\del$, it follows from
(\ref{eq:chernoffbound}) that $E\left\{ e^{g(0) N} \right\} <
\infty$. This proves (\ref{eqn:finite_mean_add_int}).
%
\end{proof}

%

\begin{center}
{\bf Conclusion}
\end{center}
In this paper, we have proposed a new power control strategy and a non-ALOHA protocol to achieve finite expected time (delay) for a packet to successfully reach its nearest neighbor with the SINR model. In prior work, it is known that the expected time for a packet to leave its source is infinite with an ALOHA protocol, severely limiting the effectiveness of such wireless networks. 
The power control strategy chooses power to completely overcome the path-loss effect towards the nearest neighbor in a defined cone that contains the destination, 
and attempts transmissions with appropriate probability so as to satisfy an average power constraint at each transmitter. In addition to achieving finite expected delay, we also show that our policy achieves non-zero information velocity, that is defined as the ratio of the successfully covered 
 distance and the delay needed to reach that distance, as time goes to infinity. As a result, packets can flow between any source-destination pair over multiple hops at a non-zero rate. Some outstanding questions that remain, are: what is the optimal choice of the angle of cone, larger the cone angle shorter is nearest neighbor distance and per-hop delay but requires more hops till the destination and vice versa, whats the best lower bound on the per-hop delay and the upper bound on the information velocity.

\begin{center}
{\bf Acknowledgments}
\end{center}

This paper has benefited from numerous useful discussions with Manjunath Krishnapur and D. Yogeshwaran. 

\bibliographystyle{IEEEtran}
\bibliography{IEEEabrv,Research}
\end{document}